\documentclass[11pt]{article}

\usepackage{amssymb, amsmath, amsthm}
\usepackage{epsfig,graphicx}
\usepackage{geometry}
\usepackage{algorithm}

\usepackage{epsfig,graphicx}

\usepackage{color-edits}
\addauthor{mb}{red}
\addauthor{lb}{blue}
\addauthor{nn}{green}

\newcommand {\ignore} [1] {}

\newcommand {\arboricity} {w}

\newtheorem{theorem}{Theorem}[section]

\newtheorem{proposition}[theorem]{Proposition}
\newtheorem{claim}[theorem]{Claim}
\newtheorem{observation}[theorem]{Observation}
\newtheorem{corollary}[theorem]{Corollary}

\title{Networks of Complements}

\author
{Moshe Babaioff\thanks{Microsoft Research, \texttt{moshe@microsoft.com}.}
	\and
	Liad Blumrosen\thanks{Hebrew University \texttt{blumrosen@gmail.com}.}
	\and
	Noam Nisan\thanks{Hebrew University and Microsoft Research, \texttt{noam@cs.huji.ac.il}.}
}

\date{}

\begin{document}

\maketitle

\begin{abstract}
We consider a network of sellers, each selling a single product, where the graph structure
represents pair-wise complementarities between products.
We study how the network structure affects
revenue and social welfare of equilibria of the pricing game between the sellers.  We prove positive and
negative results, both of ``Price of Anarchy'' and of ``Price of Stability'' type, for special families of graphs (paths, cycles)
as well as more general ones (trees, graphs).
	We describe best-reply dynamics that converge to non-trivial equilibrium in several families of graphs, and we use these dynamics to prove the existence of approximately-efficient equilibria.
\end{abstract}



\maketitle

\section{Introduction}

Sellers typically do not operate in isolated markets but in conjunction with other sellers and buyers.
In particular, sellers need to take into account the fact that buyers may value the goods they sell as substitutes or complements to goods sold by others.
This has a tremendous impact on how sellers compete with each other.
Indeed, in Cournot's \cite{Cournot1838} famous paper from 1838 about
sellers who
compete through quantities, Cournot also describes a model of a duopoly selling perfect complements, zinc and copper.
In Cournot's example, a manufacturer of zinc may observe that some of her major customers produce brass (made of zinc and copper);
 Therefore, zinc manufacturers compete not only with other zinc sellers, but they also indirectly compete with manufacturers of copper, as both target the money of brass producers.
 We are interested in a more complicated competition structure, where as zinc can be also used for Galvanization of iron, zinc sellers compete at the same time with sellers of iron.
 In a similar way, iron is also demanded by car manufacturers that  need to purchase glass from other sellers, and so on.

Another classic example is by Ellet \cite{Ell66}, who studied how owners of two consecutive segments of a canal determine the tolls for shippers; Clearly, every shipper must purchase a permit from both owners for being granted the right to cross the canal.
Another, more contemporary, example might be a high-tech or pharmaceutical firm that must use two registered patents to manufacture its product; The owners of the two patents quote prices for the usage rights, and these patents can be viewed as perfect complements for the firm. As a final example, consider an international trader who wishes to export goods from country X to country Y, and needs to pay license fees to both countries. In the last two examples, one can see how licenses have a network structure, as each patent may be needed for the production of several different products, and trade may occur between country X and country Y, but also between country X and another country Z, etc.

In this paper, we study markets where
goods are complementary to several other goods, but not substitutes.
Price discrimination is impossible, and sellers need to offer the same price in all markets.
This situation creates a global price competition between sellers, which raises interesting questions we aim to explore:
What kinds of equilibria exist in these games?
	How efficient are the equilibria in this game? Will natural dynamics reach highly efficient equilibria?	
	


The structure of the market plays a central role in our analysis. We model the market using a weighted undirected graph, where each vertex represents a seller of a certain good. An edge with weight $v$ between vertices $i$ and $j$ indicates that there is a buyer that is willing to pay an amount $v$ for the bundle of goods $\{i,j\}$.\footnote{
Throughout this paper, for the simplicity of presentation, we consider buyers that demand bundles of size 2 which we can model using graphs;
Buyers that demand bundles of arbitrary sizes can be modeled as hyper-edges on a hyper-graph. We show how some of our results extend to hyper-graphs in the appendix.
}
In the above example, a market with sellers of copper, zinc, iron and glass can be represented as a path-graph with 4 vertices (See Fig. \ref{fig:simple-networks}), where edges connect copper and zinc, zinc and iron, and iron and glass.

We study the following simultaneous full-information pricing game between sellers on a graph: the sellers observe the values of the buyers, which are common knowledge; each seller then posts a single price to all buyers;  a buyer on an edge buys the two goods on this edge if the total price (the sum of the prices of the two goods) is no larger than her value.
Sellers have zero manufacturing costs and unlimited supply, and
the profit of each seller is the price she posted times the number of buyers that accepted this price. In a (pure) Nash equilibrium, no seller would benefit from changing its price given the prices offered by the other sellers.

There are two natural benchmarks for measuring the quality of equilibria in such games. The first is the
\textit{maximum welfare},
which is the sum of values of all buyers.\footnote{Assuming 0 production costs for sellers, and payments cancel out.}
This is the welfare that would be achieved with zero prices for all goods.
	The second benchmark is the \textit{optimal monopolist revenue}, which is the optimal total revenue achievable by a monopolist that owns all the goods in the network. (Clearly, the optimal monopolist revenue is always at most the maximum welfare.)
Several papers studied the problem of computing the optimal pricing for a monopolist in our setting. The problem was proposed by \cite{Guruswami05} which showed that it is APX-hard and presented an approximation algorithm that is logarithmic in the number of buyers. A $1/4$-approximation algorithm was later presented by \cite{Balcan06} when buyers are interested in bundles of size 2, and recently \cite{Lee15} showed that this bound is tight under some computational assumptions. A similar algorithm for this problem was suggested by \cite{LBATL07} in the context of setting up peering connections in networks. 
	\cite{KMR11} provided an improved bound for monopolist revenue maximization with buyers having the same value but interested in one or two items.
Unlike these papers which focus on the monopolist's algorithmic problem,
we focus on analyzing equilibria in the game between competing sellers and on how the welfare and revenue of
these equilibria approximate the above benchmarks.

It is straightforward to see that some equilibria in these games demonstrate a complete market failure.
For example, a vector of prices of $\infty$ posted by all sellers forms a Nash equilibrium that yields zero welfare and revenue.
Yet, there might be multiple equilibria in these games, and one may  hope that other equilibria perform better.
Indeed, we prove that for some families of graphs
the best equilibria (``price of stability'') have high revenue and welfare.
%
Observe that we cannot hope for full efficiency as even for path graphs, the most efficient equilibria are sometimes not fully efficient:
For example, in Figure \ref{fig:simple-networks}, if either the seller of zinc or of iron offer a price above 1, then one of its edges will not be sold. However, there is no equilibrium in which these two sellers offer prices at most 1: if this occurs, each seller gains at most 2 (from selling to its two edges) but they can get revenue of at least 5 by offering 5 and selling to the middle edge only.


	\begin{figure}[t]
		\center
		\framebox{\includegraphics[width=2.4in, height=0.8in]{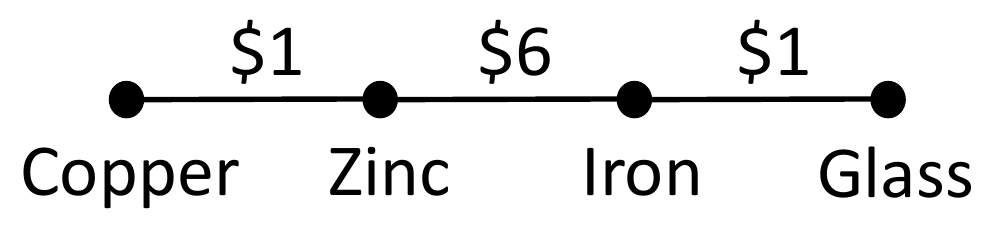}
		}
		\caption{A market with 4 sellers and 3 buyers. The weight on each edges is the amount this buyer is willing to pay for the bundle of the two adjacent goods. For example, the middle edge represents a buyer that is willing to pay $\$6$ for a bundle of zinc and iron, but has value 0 for each good alone.
\label{fig:simple-networks}}
	\end{figure}

	Another way to alleviate the problem of
		low welfare and revenue of some equilibria 
	is by restricting attention to equilibria with natural properties.
For instance, it may be unreasonable for a seller that does not sell any good to insist on a very high price despite having zero costs;
Such behavior might be considered as a malicious behavior towards his neighboring sellers. 
We therefore sometimes restrict attention to \emph{non-malicious} equilibria,
where all sellers that are not selling at all offer a price of 0.
	Non-malicious equilibria  serve as a main tool for proving the existence of approximately-optimal equilibria: we show that the revenue of every such equilibrium approximates the maximum welfare, and we show cases where such equilibria exist.



\vspace{2mm}

\noindent \textbf{Our Results:}\\ Our main goal is to understand how the network structure affects various properties of equilibria in markets.
We aim to understand how well can the welfare and revenue in equilibrium approximate the maximum welfare 
and the optimal revenue. 
We explore whether best-response dynamics can lead to approximately-optimal equilibria, and we prove the existence of such equilibria in several settings.
We study graphs of different complexities, like paths, cycles, trees and general graphs.

We start by considering some special families of graphs (paths and cycles) and show that there is a  natural best response dynamics that reaches equilibria with a constant fraction of the optimal welfare as revenue.
For path graphs, we consider a dynamics starting with all sellers pricing at infinity. Then, we change the price of a seller on one side of the path to zero,
and let each seller in turn best reply to current prices.
We show that when this best-reply process reaches the other end of the path, it ends up in an equilibrium.
Moreover, for at least one of the end points, the revenue when starting from that point is at least half the maximum welfare.
An extension of this dynamics for cycles
converges to an equilibrium with revenue of at least one quarter of the optimal welfare.

After we prove that for some families of graphs there exist Nash equilibria
with good revenue guarantees, the immediate question is whether this generalizes to more complex graphs.
It turns out that non-maliciousness of equilibria can be used to prove bounds on the revenue. 
	In particular, we show that for any tree, every non malicious equilibrium has revenue that is at most $O(\log d)$ factor  smaller than the optimal welfare, where $d$ is the maximum degree. 
	One might hope that such a result can be achieved for general graphs as well, but unfortunately this is not the case. For a
	clique of degree $d$, the loss can be linear in $d$,\footnote{Consider a clique of degree $d$ with all edges with the same value of $1$.
		Such a clique has a non-malicious equilibrium in which one seller prices at $0$ while all others price at $1$.
		The welfare in this non-malicious equilibrium is only linear in $d$, while the optimal welfare is quadratic in $d$, and thus there is a linear loss in $d$.} and therefore we cannot hope for a logarithmic fraction 
	for general graphs.
	Our main positive result presents a refined bound that depends not only on the maximum degree, but also on the \emph{arboricity} of the graph (see, e.g., \cite{NW61}). This additional parameter nicely captures the difference between trees and cliques (and other graphs).
	The arboricity of an undirected graph is the minimum number of forests into which its edges can be partitioned. (For example, a tree has arboricity 1, a cycle has arboricity 2, and a clique of $n$ nodes has arboricity $n/2$.)
We prove the following bound on the revenue of any non-malicious equilibrium:\footnote{
We note that the theorem holds also for non-simple graphs that have parallel edges.
}

%


\vspace{2mm}

\noindent \textbf{Theorem:}	\emph{In every graph
	with maximum degree $d$ and arboricity $w$ and every non-malicious equilibrium in it, the total revenue of all sellers is
	at least $\Omega(\frac{1}{w+\log d})$ 	 	
	of the maximum
	welfare.}


%

\vspace{2mm}

\ignore{ 
It is important to note that the approximation guarantee in this result is not solely a function of the degree of the graph. For example, consider a tree with a maximum degree $d$ and a clique over $d$ vertices. These graphs have the same degree, but our theorem shows that the tree exhibits considerably better approximation that is only logarithmic in the maximum degree, as its arboricity is only 1 (compared to a linear arboricity of the clique). }
\ignore{
We  note 
that the above result is not restricted to simple graphs,\footnote{This result also holds when some buyers are interested in a single good, not a pair.}  but holds in graphs with parallel edges as well.
Parallel edges can affect both the degree of the graph and its arboricity.
Nevertheless, our model is described for simple graphs for the simplicity of notation. Note that our impossibility results (see below) hold even under the restriction to simple graphs and do not require parallel edges.
}

{
	The above theorem does not claim anything about the existence of non-malicious equilibria, but only bounds the revenue obtained by such equilibria if they exist.\ignore{ 
		This result implies that for any network for which a non-malicious equilibrium exists, there is an equilibrium with revenue which is at least some fraction of the welfare, and this fraction is given as a function of the maximum degree and the arboricity. }
For this result to imply a bound on the ``price of stability'' in such games,\footnote{Our result above also has a ``price-of-anarchy'' flavor, as we prove that all non-malicious equilibria (when exist) exhibit the approximation guarantee.
} one needs to show that non-malicious equilibria actually exist.
We prove the existence of non-malicious equilibria via the natural heuristic of repeated best-reply dynamics.
Such a dynamics starts with arbitrary prices; At each step, a seller who is not best replying is chosen, and he updates his price to a best reply (breaking ties non-maliciously - towards $0$ price).
Our main result in this context shows that in tree graphs, a specific sequence of best replies does stop at a non-malicious equilibrium.


\vspace{2mm}

\noindent \textbf{Theorem:}	\emph{In every tree, for every initial profile of prices, there exists a sequence of best replies that terminates in a non-malicious Nash equilibrium.
}

\vspace{2mm}


In particular, this result implies the existence of non-malicious equilibria in trees, and together with the approximation theorem above, we conclude that the price of stability in trees (i.e., the approximation achieved by the best equilibrium in every tree) is at least $\Omega(\frac{1}{\log d})$.
Of course, it would be interesting to strengthen this result and show that for any graph such dynamics terminate in a (non-malicious) equilibrium. Based on simulations we have executed, we conjecture that this is indeed the case.


\vspace{2mm}

\noindent \textbf{Conjecture:} 	\emph{In every graph any non-malicious best response dynamics starting from any price vector converges to an equilibrium in polynomial number of steps.}

\vspace{2mm}

This is the strongest version of our conjecture. One can also try to prove weaker versions of the conjecture by restricting the graph structure, order of plays, and relaxing the required time to convergence.



}


We also prove several impossibility results.
We first show that the $\Omega(\frac{1}{\log d})$ approximation for trees is tight, by presenting trees in which the welfare of every non-malicious equilibrium is 
at most an $\Omega(\frac{1}{\log d})$ fraction of the monopolist revenue (and thus of the optimal welfare).
However, this bound is proved using constructions for which there are other (``malicious'') equilibria that achieve a constant approximation.
We therefore strengthen this impossibility result and show instances where {\em all} equilibria, malicious or not, achieve bad results.
In particular, we would like to find out whether the revenue of the best equilibrium is always a constant fraction of the monopolist's revenue.
Our main lower bounds show that the answer is negative, even for trees.
%
%
%


\vspace{2mm}

\noindent \textbf{Theorem:}	
{ \em For graphs with maximum degree $d$:}
\vspace{-1mm}

    $\bullet$ { \em There are {\em graphs} for which the welfare 
    	in every Nash equilibrium is at most $O(\frac{1}{\sqrt{\log d}})$ fraction of the monopolist's revenue. }
			    		
		

\vspace{-1mm}

$\bullet$	{ \em There are {\em trees} for which the welfare 
			in every Nash equilibrium is at most $O(\frac{1}{\log \log d})$ fraction of the monopolist's revenue.}


\vspace{2mm}


%
%
%
%
%

Note that the fact that the revenue in every equilibrium is low compared to the monopolist revenue, implies that it is also low compared to the maximum welfare. Moreover, our bounds are actually stronger, showing that not only the revenue in equilibrium is low, but also the welfare.

\vspace{2mm}


\noindent \textbf{Related Literature:}\\
The analysis of price competition goes back to Cournot \cite{Cournot1838} and Bertrand \cite{Bert83}.
The famous competition model of
\cite{Cournot1838} describes producers who compete through quantities, but in the same work he also described a model of a duopoly selling perfect complements.

\cite{LOS:J} consider buyers interested in perfect complements,  where each bidder is interested in one specific bundle (single-minded bidders);
They study mechanism design for a single seller and buyers with private information, where our focus is on competition between multiple sellers with complete information. Auctions for networks of buyers and sellers with private information were studied in \cite{KM01}.

There is a line of literature studying interactions of sellers and buyers over networks. Among them, \cite{KT08} considered bargaining in networks, where agents can choose whom they want to negotiate with and the solution implies matching of buyers and sellers.  In our model, neighbouring sellers also ``negotiate'' on dividing the value of the buyer on the edge between them, but unlike \cite{KT08},
a seller
can serve several neighbouring buyers simultaneously (with the constraint of non-discriminatory pricing). Sub-game perfect equilibria in bargaining games were studied in \cite{Cor04}.
\cite{KKOPS05} studied general equilibrium models on networks with linear utilities of buyers.

\cite{BabaioffLN13} considered a similar model of sellers that compete in prices, and a trade network that is modeled as graph, with one main difference from this paper: items are substitutes for the buyers.
Namely, a buyer is interested in buying {\em either}  from one seller or from the other seller on the edge, while here we consider buyers that are interested in {\em both}.
This difference implies significant differences in results. For example, there, unlike our paper, pure Nash equilibria hardly ever exist.
For trees, equilibrium utilities of the sellers are uniquely defined, while here they are not.\footnote{ 
Our framework can be also viewed as a variant of graphical games \cite{KLS01}.
There are several known algorithms for computing equilibria in graphical games (see survey \cite{Kea07}).
Another related family of games is Polymatrix games (see \cite{CD11,DFSS14,CDT09}), where an action of a player is played in several simultaneous bimatrix games (the games in our paper have continuous action spaces); This class of games played an important role in showing hardness of equilibrium computation results.
Our goal in this paper is to find equilibria with good economic properties, and compare their properties to the optimal non-strategic solution.
Our paper shows that allowing losing sellers to price ``maliciously'' might be beneficial to society as a whole, this phenomenon was termed the ``Windfall of Malice'' by \cite{BabaioffKP09}.
}

\section{Model}
\label{sec:model}

We study trade networks which are modeled as weighted undirected graphs,
where sellers are represented by nodes and buyers by weighted edges.
The set of edges
is denoted by $E$,
let $N$ be the set of nodes and let $i,j\in N$ denote generic sellers.
Sellers have no costs and can supply any number of items. The sellers are the players in the game, interested in maximizing revenue. Each buyer is single minded and is interested in the bundle of items sold by both sellers that lie on  that edge. The weight $v_{i,j}$ on the edge represents the value of this buyer for buying the bundle $\{i,j\}$. In particular, buyers gain zero value from buying each item alone.
A seller $i$ posts a single price $p_i\geq 0$ that will be available to all incident edges (cannot price discriminate between buyers).
A buyer of bundle $\{i,j\}$ buys if and only if $v_{i,j}\geq p_i+p_j$.

For a given price vector $p$ we denote the set of sold edges by $S(p)$.
The edge $(i,j)$ is {\em tight} if $p_i+p_j= v_{i,j}$ (the sum of prices of the sellers on the edge equals the value of the edge).
We say that {\em an edge $(i,j)$ has slack} if $p_i+p_j<v_{i,j}$.
For a given seller $i$ and an edge $(i,j)$, {\em the slack of seller $i$ on that edge} is
$v_{i,j}-p_j$ (that is, it equals to the difference between the value of the edge and the price of the \emph{other} seller on that edge).

For a given price vector $p$, the {\em revenue} $r_i$ of seller $i$ with price $p_i$ is $p_i\cdot |\{j | (i,j)\in S(p) \}|$, that is, $p_i$ for every adjacent edge that is sold.
The {\em total revenue} is $\sum_i r_i$ and the {\em welfare} is $\sum_{(i,j)\in S(p)} v_{i,j}$.

The sellers compete in a game in which they simultaneously post prices. These prices form a {\em (pure) Nash equilibrium (NE)} if each seller maximizes his revenue given the prices of all other sellers. In this paper we only consider pure Nash equilibria.
\ignore{
\mbcomment{I have moved the observation to the appendix.}
\mbdelete{The following is a useful simple observation about the price set by a best responding seller, it is used repeatedly throughout the paper.
\begin{observation}
	\label{obs:tight-edge}
	For any graph and any vector of prices for the sellers, if seller $i$ is best responding and his revenue is positive, then there is an incident edge that is tight.
\end{observation}

To see this, note that if no edge is tight, a slight increase in price will increase the seller revenue (some edges must be sold as the seller's revenue is positive).
}
} 

For a given network, the {\em revenue of the monopolist} is the supremum of the total revenue over all price vectors for the sellers. 
The {\em maximum welfare} is simply the sum of values of all edges (this can be obtained when every seller prices at $0$).
We study the ratio between the revenue in equilibrium and either the revenue of the monopolist or the maximum welfare (whichever bound is harder to prove):
For our positive results we show that the equilibrium revenue is some fraction of the maximum welfare (which clearly implies  approximation to the monopolist revenue),
while for our negative results we prove inapproximability of 
the monopolist revenue (which implies the same result with respect to the maximum welfare).

\vspace{2mm}

\noindent \textbf{Non Malicious Equilibria:}
We have already observed that the welfare of the worst equilibrium is arbitrarily low.
Thus, we sometimes study an equilibrium refinement in which sellers with $0$ utility price at $0$.
We say that a losing seller (seller with $0$ utility) is {\em non malicious}  if she prices at $0$.
A price vector for the sellers is called {\em non malicious} if every losing seller is non malicious.
This, in particular, implies that every seller's revenue is at least as high as the price he sets.
We use the concept of non-malicious NE when proving our main result: we first show that such equilibria always approximate the optimal welfare, and we then show that they always exist in trees. This implies a positive result on the price-of-stability in trees.

\section{Special Families of Graphs: Paths and Cycles}
\label{sec:compute}

In this section, we present positive results for special cases of graphs and algorithms that compute Nash equilibria with high revenue for those cases. These algorithms consist of a particular sequence of best responses of the sellers.
We 
present a linear time algorithm for path graphs that finds an equilibrium that obtains revenue that is at least half of the maximum welfare.
We use it to present
a linear time algorithm for cycles that 
computes an equilibrium with revenue at least a quarter of the maximum welfare.

\subsection{Path Graphs}

We first consider the simplest non-trivial graphs: path graphs.
Our algorithm (Algorithm 1) assumes that sellers are indexed from left to right, with the leftmost seller indexed by $1$.
The algorithm starts with all sellers pricing at $\infty$ and then changes the price of the leftmost seller to $0$.
Then, it goes over the sellers from left to right, letting each seller to choose a best response. If the seller cannot get positive utility, he prices at the value of the previous edge (not at 0 - and thus the equilibrium might be malicious).  
The way we break ties for sellers with $0$ utility ensures that
after seller $i$ best responds,
all previous sellers are still best responding (as seller $i$ leaves no slack to the previous seller).
Our algorithm may achieve terrible revenue when executed in a certain direction on the path (for example, when the weights of the edges are monotonically decreasing), but in such cases we show that running it from for the opposite direction will perform well.
This is essentially since every time the value of the edge is larger than the previous edge, it is sold tightly.
The algorithm thus picks the better equilibrium out of executing the above procedure starting from one end, or from the other.\footnote{We illustrate the algorithm for the path graph of Fig. \ref{fig:simple-networks}: 
The algorithm chooses a price of $0$ for the copper seller, then a price of $1$ for the zinc seller, a price of 5 for iron and a price of 1 for glass. This is a NE with revenue 7 (out of maximum welfare of 8).
By symmetry, the revenue starting from the other end would be the same.
} 

{ 
\begin{algorithm}
{  \small 
		\begin{enumerate}
		\item Initialize all prices to infinity, and the price of the first (leftmost) seller to $0$.
		\item Starting from the second seller, go over the sellers from left to right and let each seller best respond to the current prices.
		If every price of the current seller $i$ gets him $0$ utility,
set his price to the value of edge $(i-1,i)$ (the value of the previous edge).
		\item Mirror the path left to right (mapping node $i$ to become node $n+1-i$) and repeat the above algorithm. Output one of these two price vectors with the higher total revenue. 
	\end{enumerate}
	\caption{A $\frac{1}{2}$-Approximation for path graphs.}
	\label{alg:line-ne}
}
\end{algorithm}

This algorithm terminates in a NE with revenue of at least half  the maximum welfare (proof in Appendix \ref{app:lines}):


\begin{proposition}
\label{prop:line-half}	
	For any path graph, Algorithm \ref{alg:line-ne} terminates in a NE after a linear number of steps. The total revenue in this equilibrium is at least half of the maximum welfare.
\end{proposition}

\subsection{Cycles}

	Our next result shows that in any cycle graph there is always a NE with high revenue. Moreover, we show how to compute such a NE in linear time; The algorithm runs the above path-graph algorithm for one full round on all sellers, and then initiates a best-response dynamic that will end up in such an equilibrium.


Fix an arbitrary seller and mark him as seller $1$. Going clockwise along the cycle starting from seller $1$, number the other sellers  $2$ to $n$.
We present the algorithm as Algorithm \ref{alg:cycle-ne}.

\begin{algorithm}
{  \small 
		
	\begin{enumerate}
		\item Initialize all prices to be infinity, initialize the price of seller $1$ to $0$.
		\item Going clockwise: Starting from the second seller till the first seller (inclusive), go over the sellers clockwise and let each seller best response to the current prices.
	     If every price of the considered seller $i$ gets him $0$ utility\ignore{(no price gives positive utility)}, set his price to the value of edge $(i-1,i)$ (the value of the previous edge in the cycle, with edge $(n,1)$ being the edge prior to $1$).
		\item Run a best response dynamics (by iteratively letting a seller that is not best responding to update his price) till it stops. 
		\item Mirror the cycle and repeat the above algorithm (i.e. the new order after seller $1$ is $n,n-1,\dots,2,1$). Output one of these two price vectors with the higher total revenue.
	\end{enumerate}
	\caption{A $\frac{1}{4}$-Approximation for Cycles}
	\label{alg:cycle-ne}
}
\end{algorithm}

\begin{theorem}
	\label{thm:cycle}
	For any cycle, Algorithm \ref{alg:cycle-ne} terminates in a Nash equilibrium after linear number of steps. The total revenue in this equilibrium is at least one quarter of the maximum  welfare.
		
\end{theorem}
We defer the details of the proof to Appendix \ref{app:cycle}, and we next present an informal description of the algorithm together with an outline of the  proof.
First, the algorithm essentially runs Algorithm \ref{alg:line-ne} in one direction, starting with a price of 0 for seller $1$ and going clockwise, letting each seller best respond (pricing at the value of the prior edge if his utility is 0) including the first seller 
as the last seller to change his price. We note that the first seller is now best responding (as the edge to the second seller had no slack, so he could not get any utility out of it).
After that cycle completes, we argue that there is at most one seller (the second seller) that is not best responding.
From this point on, the algorithm runs a best response dynamics till it stops.
We argue that at each point, the only seller that might not be best responding is the next seller, and once he best responds, either the sellers are in a NE or the edge with the previous seller is tight. As we show, it follows that this dynamics will stop after a linear number of steps. We also show that the algorithm running either clockwise or counter-clockwise gets revenue that is at least one quarter of the maximum welfare: the path-graph algorithm gets at least half, and from then edges are only added (and sold tightly), except the edge following the last seller to update his price that might be dropped. We show that the revenue of the dropped edge is no larger than the remaining revenue,
thus we lose at most half the revenue obtained by the path algorithm.

\section{A Positive Result} 
We start by presenting our main positive result, showing that the revenue in some NE is at least some fraction of the maximum welfare (and thus also some fraction of the monopolist revenue). We prove this claim by showing that this is true for any non-malicious NE.
For any network for which a non-malicious equilibrium exists (like in trees, see Theorem \ref{thm:best-res-dyn-tree}), this implies that at least the same fraction can be obtained in the highest revenue equilibrium.\footnote{
This result actually holds for a more general graph model, where some buyers are interested in a single good, not a pair, and also for graphs that contain parallel edges (that is, there might be multiple buyers interested in each pair of goods).
}

The bound we present is in terms of the maximum degree and the  arboricity of the graph.
The {\em maximum degree} is the largest degree of any node: the maximal number of buyers that demand an item from the same seller.
The {\em arboricity} of an undirected graph is the minimum number of forests into which its edges can be partitioned.
Equivalently, it is the minimum number of spanning forests needed to cover all the edges of the graph.

To gain some intuition for the notion of arboricity, here are some simple examples. A forest (and a tree)
has arboricity $1$.
A cycle has arboricity $2$ as it clearly cannot be spanned by a single forest, 
but it can be partitioned into two trees, a spanning tree and the missing edge.
A clique of size $n$ has arboricity $\lceil n/2 \rceil$: the arboricity is at least $n/2$ as any tree has at most $n-1$ edges and the clique has $n(n-1)/2$ edges, and it is at most $\lceil n/2 \rceil$ as it can be covered by $\lfloor n/2 \rfloor$ Hamiltonian paths, plus a star for odd $n$.

\begin{theorem}
	\label{thm:graph-poa}
	In every graph with maximum degree $d$ and arboricity $w$ and every
non-malicious NE in it, the total revenue of all sellers is
	at least $\frac{1}{2(w+1+\ln d)}$ 	
	fraction of the maximum welfare.
\end{theorem}
\begin{proof}
	Fix a graph and a non malicious Nash equilibrium in it.
	We will say that a vertex $u$ is {\em low for the edge $(v,u)$} if $u$'s price in the equilibrium is at most half of the value of the edge $(v,u)$.
	We will say that an edge is {\em all-high} if both of the vertices on it are not low for the edge (in this case clearly the edge does not buy).
	Let $E^H$ denote the set of all-high edges. For a vertex $v$ let $E^v$ denote the set of edges $(v,u)$ (edges adjacent to $v$) such that $u$ is low for $(v,u)$.
	Observe that since every edge is either all-high or low for some node, 	the set of all edges $E$ is covered as follows: $E=E^H \cup (\cup_{v} E^v)$.
	Thus $$\sum_{(v,u)\in E} v_{(v,u) }\leq  \sum_{(v,u)\in E^H} v_{(v,u) } + \sum_{v} \sum_{(v,u)\in E^v} v_{(v,u)} $$
	Denote the total revenue by $r$. To complete the proof we bound each of the two terms separately, the first by $2w\cdot r$ and the second by  $2(\ln d +1)\cdot r$,
	together proving the theorem.

	We start by bounding the edges that are not all-high. Let $r_v$ denote the revenue of $v$. This claim is well known (e.g., \cite{GHKSW06}) and we present the proof for completeness in Appendix \ref{app:pos-proofs}.\footnote{
The claim essentially says that a single price can gain revenue of at least a logarithmic fraction of the total demand (in our case, it is the residual demand given the prices of the neighbours- which is at least half the value of each edge as these edges are not all-high.).
}
\begin{claim}
\label{claim:log-approx-demand-curve}
		For every node $v$ it holds that $\sum_{(v,u)\in E^v} v_{(v,u)} \leq r_v \cdot 2(\ln d +1) $
\end{claim}
By the claim we derive that
	$$\sum_{v} \sum_{(v,u)\in E^v} v_{(v,u)} \leq \sum_{v} 	r_v \cdot 2(\ln d +1)  = 2(\ln d +1)\cdot r$$
	We next bound the total value of the all-high edges.
	To take care of the total weight of the all-high edges
	we will use the fact that in a graph of arboricity $w$ there exists a mapping from edges to vertices such that every edge is mapped to one of its two vertices and no vertex has more than $w$ edges mapped to it (this is true as we can just root each tree and map every edge to its child node).
	Since the edge is all-high and the equilibrium is non malicious, the price -- and thus also the revenue -- of each of the two vertices on the edge is at least half the value  of the edge.
	Summing, again, over all vertices, we get that the
	total weight of all all-high edges is at most $2w$ the total revenue of all vertices.
	Stated formally, consider the mapping $M$ from $E^H$ to the set of nodes $N$ that maps each edge to an adjacent node and never maps more than $w$ edges to the same node. For every node $v$ that incident at least one high edge, define $u^*(v)$ to be a vertex such that for all $u$ such that $(v,u)\in E^H $ we have
$v_{(v,u^*(v))}\geq v_{(v,u)}$. 
	Then, it holds that
\begin{align}
	\sum_{(v,u)\in E^H} v_{(v,u) } & \leq
	\sum_v \sum_{u|M(v,u)=v} v_{(v,u) } \leq
	\sum_{v | \exists u \; M(v,u)=v} 	w\cdot v_{(v,u^*(v))} \\
     & \leq
	w\cdot \sum_v 2\cdot r_v \leq 2w\cdot r
\end{align}
\end{proof}

	For any network for which  a non-malicious equilibrium exists, the theorem ensures that some Nash equilibrium has high revenue. Thus, it can be viewed as a ``Price of Stability'' result for such networks.
In particular, it bounds the price-of-stability for trees, as for any tree a non-malicious equilibrium exists by Theorem \ref{thm:best-res-dyn-tree}.
It also ensures that every non-malicious equilibrium has high revenue, thus can be viewed as a ``Price of Anarchy'' result for non-malicious equilibria.

The theorem implies a bound on trees that is only logarithmic in the maximum degree.
\begin{corollary}
	\label{cor:tree-lb}
	In any tree with maximum degree $d$ and every non-malicious equilibrium in it, the total revenue of all sellers is $\Omega(1/\ln d)$
	fraction of the maximum welfare.
\end{corollary}

\section{Impossibility Results}
\label{sec:impossibilities}
Theorem \ref{thm:graph-poa} gives a positive result, ensuring that the revenue in equilibrium is some fraction of the welfare.
Yet, it might be possible that an improved bound can be shown.
\ignore{
While equilibria that does not satisfy the non-malicious condition  might have arbitrarily bad revenue, the best revenue equilibrium might have high revenue.

Therefore, one might like to understand the best equilibrium revenue both when not restricting to non-malicious NE, and when restricting to this smaller set.

\subsection{Arbitrary Nash equilibrium}
Proposition \ref{prop:trees-non-mal-lb} demonstrated that the highest revenue (unrestricted) equilibrium might have much higher revenue than the best non-malicious equilibrium.
As the set of all equilibria is much larger than the set of non malicious NE, one might hope that some equilibrium is always of high revenue.
}
Specifically, we would like to answer the following question: Is the revenue of best equilibrium always a constant fraction of the revenue of the monopolist?
Unfortunately, the answer is no, even for trees.
We first present a lower bound 
for general graphs. This bound is slightly weaker than the logarithmic bound (in the maximum degree)
of Theorem \ref{thm:graph-poa}.
We then present a lower bound for trees, showing that
the best equilibrium revenue is not necessarily a constant fraction of the monopolist revenue.
For missing proofs see Appendix \ref{app:lower-bounds}.

\subsection{General Graphs}
We start with the lower bound for general graphs. For proving the theorem, we construct graphs with arboricity that is much smaller than their maximum degree.
\begin{theorem}
	\label{thm:general-lb}
	There exists  a family of graphs with maximum degree $d$ and arboricity $w$ for which $w^2=O(\ln d)$ and the revenue that a monopolist seller can get is factor $\Omega(w)$ larger than the revenue (and welfare) in any Nash equilibrium. It terms of $d$, the factor can be as large as $\Omega(\sqrt{\ln d})$ when $w^2=\Theta(\ln d)$. 
\end{theorem}

The bound is proven using the following construction (see Fig. \ref{fig:graph-lb}).
There is a clique of size $w+1$ and any node in the clique is connected to an ``harmonic gadget'': $d$ edges with values $1,1/2,1/3,$ $...,1/d$.
The value of an edge connecting two nodes in the clique is $4$.
We first claim that for these parameters, a monopolist would price all the clique nodes at $0$ and get full revenue from all the harmonic gadgets.
In any NE, however, at most one seller prices below $1/w$, thus not much revenue is gained from the harmonic gadgets.

	\begin{figure}[t]
		\center
		\framebox{\includegraphics[width=2.5in, height=1.5in]{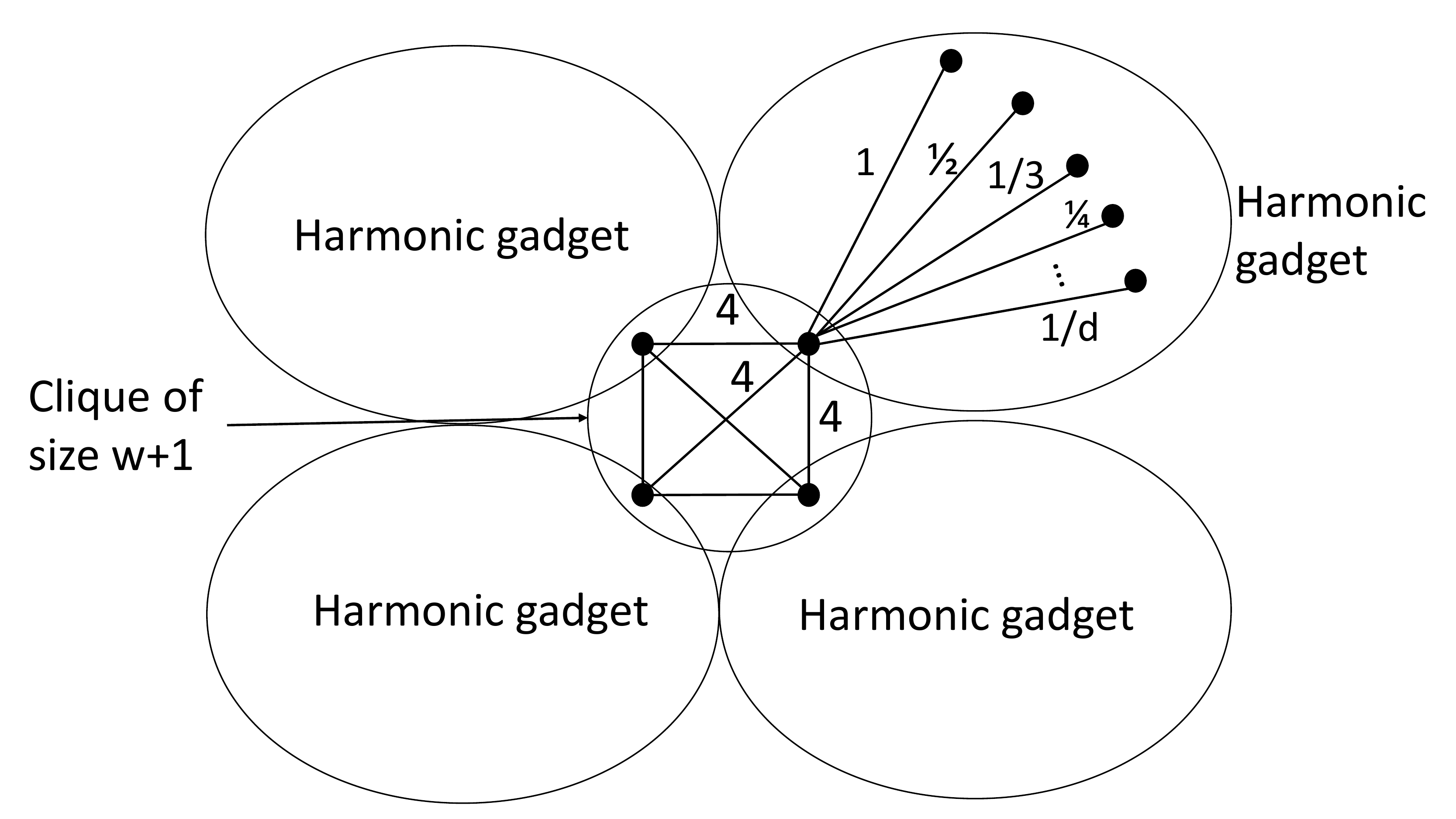}
		}
		\caption{The construction of Theorem \ref{thm:general-lb}. A clique of $w+1$ sellers, each edge in it is of value $4$.
			Each node in the clique is connected to an ``harmonic gadget''.
\label{fig:graph-lb}}
	\end{figure}
\subsection{Trees}

For trees, the above bound (when $w=1$) is a constant. 
We next present a different construction that shows that for some trees, the revenue of a monopolist can be factor $\Omega(\ln \ln d)$ larger than the revenue of the best equilibrium. In particular, it show that a constant upper bound is impossible.

\begin{theorem}
	\label{thm:general-tree-lb}
	There exists a family of trees with maximum degree $d$ for which the revenue that a monopolist seller gets is factor $\Omega(\ln \ln d)$ larger than the revenue (and welfare) in any NE.
\end{theorem}

For proving this impossibility result we construct the following graph, fixing an integer $m$ (to be determined later).
Consider a path with $2m+1$ edges. 
The first edge has value $5$. For any $j=1,2,\ldots m$, given that edge $2j-1$ has value $v$, set the value of the edge $2j$ to be $2v+2$ and the value of edge $2j+1$ to be $2v+6$. Any node of even index is additionally connected to an ``harmonic gadget'' with $d-2$ spikes: $d-2$ edges with values $1,1/2,1/3,1/4,...,1/(d-2)$. See Figure \ref{fig:general-tree-lb-prices} for illustration.

We first argue that in any Nash equilibrium there is at most a single seller of even index on the path with price that is at most $1$.
We then argue that a monopolist seller can get revenue of at least $m\ln (d-2)$ by pricing every seller on the path at $0$ and gaining all the revenue from the $m$ harmonic gadgets.
We conclude that the ratio of the monopolist revenue to equilibrium revenue (and welfare) is at least $ \frac{\max\{m \ln (d-2), 2^m\}}{3+\ln d + 2^{m+4}}. $
For $m=\Theta(\ln \ln d)$ this ratio tends to $\Omega(m)$ as we aimed to prove. Full details appear in Appendix \ref{app:lower-bounds}.

	For trees, there is still a gap between this lower bound of $\Omega(\ln \ln d)$ and the upper bound of $O(\ln d)$ for non-malicious NE as implied by Corollary \ref{cor:tree-lb}. Closing this gap is left as an open problem.

\subsection{Non malicious Nash equilibrium}

We have observed that the worst Nash equilibrium might have zero revenue, even for a single edge. In this section, we consider non-malicious NE and present two simple examples that show that the upper bound of Theorem \ref{thm:graph-poa} is tight in both parameters for the non-malicious NE with the worst revenue (``Price of Anarchy'').

\begin{proposition}
	the following holds:
    \begin{itemize}
		\item For any $w\geq 2$ there exists a graph (symmetric clique) with $2w$ nodes and arboricity $w$ for which the revenue in some non-malicious equilibrium is smaller than the monopolist revenue by a factor of at least $w$. 
Moreover, for that graph, every best-response dynamics starting at zero prices converges to such an equilibrium.
		\item For any $d\geq 2$ there exists a tree (star) with maximum degree $d$ for which the revenue in some non-malicious NE is smaller than the monopolist revenue by a factor of at least $\ln d$.
	\end{itemize}
\end{proposition}


The two claims directly follow from the following simple examples.
		Consider first the clique graph of $n=2w$ nodes with all weights being $1$. Note that such a clique has arboricity $w$.
		Here is a non-malicious Nash equilibrium: one of the sellers prices at $0$ and all others price at $1$.  The total revenue (and total welfare in this NE) is exactly $(n-1)$.
		Note that there exists a fully-efficient fully-revenue-extracting non-malicious equilibrium where each seller prices at $1/2$.
		In this NE the revenue (and welfare) is exactly $n(n-1)/2$. Thus, the ratio
 between the best non-malicious equilibrium revenue and the worst non-malicious equilibrium revenue is  $n/2 = w$. 
Finally, observe that in any best response dynamics starting at $0$, sellers keep increasing prices from $0$ to $1$, except when there is only one seller pricing at $0$, and the others at $1$. Thus, any  best response dynamics starting at 0 prices will end at a non-malicious NE    in which one seller prices at $0$ and all others price at $1$.	

		Consider now a star with $d$ spikes, with edge $i$ of value $1/i$.
		There is a non-malicious NE with revenue of $1$: the center prices at $1$ and all other price at 0. Revenue (and welfare) of $\sum_{i=1}^d 1/i\geq \ln d$
		can be achieved by pricing the center at 0 and any other seller at the price of its edge. The second claim follows.

\section{Best Reply Dynamics}

Our main positive result (Theorem \ref{thm:graph-poa}) ensures that for any graph for which a non-malicious equilibrium exists, there is an equilibrium with high revenue. Thus, one naturally wonders if 
all graphs admit a non-malicious equilibrium.
If so, can one find such an equilibrium in polynomial time? Is there a natural dynamics that ends in such an equilibrium?
	
A natural procedure for converging to an equilibrium is repeated best-reply dynamics, and one might hope that such dynamics will indeed always converge to a non-malicious equilibrium in polynomial time.
Such a dynamics starts with arbitrary prices, and at each step a seller that is not best replying is chosen, and he updates his price to a best reply. If the goal is finding a non malicious equilibrium, then one needs to consider a seller with $0$ utility that is not pricing at $0$ as a seller that is not best replying.
We conjecture that for any graph, such a process terminates 
in an equilibrium (which is clearly non malicious).
A stronger conjecture is that such an equilibrium will be reached in polynomial time.\footnote{We remark that we have simulated best response dynamics on general graphs with arbitrary order, and they always terminated in equilibrium very fast.}



In this section, we prove the existence of (pure-strategy) non-malicious equilibria in \emph{trees} by presenting some best-reply dynamic that converges to equilibrium in a finite number of iterations.\footnote{Our proof only ensures that the dynamics will stop, but does not ensure terminating in polynomial time.
}

%
%

\begin{theorem}
	\label{thm:best-res-dyn-tree}
In every tree, for every initial profile of prices, there exists a sequence of player best replies that terminates in a non-malicious Nash equilibrium.
\end{theorem}

We defer the complete proof to Appendix \ref{app:best-response-trees}.
We now present a sketch of the proof,
and recursively define the sequence of best responses.
We  pick a leaf $u$ of the tree that is connected to the rest of the tree via vertex $v$. 
	Let $x_0$ be the initial price of $u$ and $y_0$
	be the initial price of $v$.  Our best reply dynamics will proceed by repeatedly, for $i=1....$, let $x_i$ be $u$'s reply to $y_{i-1}$, and then recursively
	use a best reply sequence that updates the rest of the tree, assuming that $u$'s price is set to $x_i$.  Note that this recursive best-reply sequence
	starts with $v$ updating his price (since all the other vertices are already best-replying from the previous recursive call, but then other vertices may
	update their price and $v$ may update again, and so on).  The
$v$'s price at the end of the recursive call is called $y_i$.  The
	recursive call on the sub-tree terminates due to an inductive use of the theorem (i.e., the theorem is proved by induction on the number of vertices in
	the tree).  To ensure that the induction hypothesis applies to the recursive call which is applied not just to a subtree, but rather to a subtree to which an
	extra leaf $u$ with a {\em fixed} value $x_i$ is attached, we prove the theorem (inductively) also for trees in which each vertex may have an arbitrary number of leaves with a fixed value attached to them. 


\section{Additional Results on Malicious vs. Non-Malicious Equilibria}

We give some comparative statics results regarding non-malicious equilibria; These result are deferred to Appendix \ref{sec:windfall} due to lack of space. We show that allowing sellers to be ``malicious'' might actually improve the situation both for any individual seller, and for the group of sellers (and society) as a whole.

\label{sec:windfall}


\subsection{Individual seller's utility}

We first consider the utility of any individual seller.
We say that some level of utility is {\em feasible} for a seller, if given that all other sellers price at $0$, the seller can get that utility.
We next show that for very simple scenarios, some feasible utilities cannot ever be reached in a non malicious equilibrium.
Moreover, in this example, a seller with positive feasible utilities has zero utility in every non-malicious equilibrium.

\begin{observation}
	\label{obs:zero-any-non-mal}
	For some path graph with three edges, a seller with positive feasible utilities has $0$ utility in any non-malicious equilibrium.
\end{observation}
\begin{proof}
	Consider the path with 4 sellers denoted by order as $A,B,C,D$ and edges of value $6,9,1$, in the corresponding order.
	Observe that every utility in $[0,1]$ is feasible for $D$. We next show that in any non-malicious NE his utility is $0$.
	
	We argue that the edge of value $6$ must be sold. If $B$'s price is at most $6$, it will be sold.
	If $B$'s price is larger than $6$, then by the non-malicious condition $A$'s price is $0$ and $C$'s price is at most $3$.
	In this case $B$'s utility is at most $9$, but price of $6$ will give him a higher utility of $12$.
	We conclude that $B$'s price is at most $6$. This means that $C$'s utility is at least $3$, thus his price is larger than $1$ and the edge of value $1$ is not sold, and $D$ has $0$ utility.
\end{proof}

What if we allow the equilibrium not to satisfy the non malicious condition? We next show that in this case the situation is very different:
{\em any} utility that is feasible for the seller can arise as the seller's equilibrium utility in \textit{some} equilibrium.
Moreover, for {\em any} price that the seller might set, there is always an equilibrium with the seller setting this price.
Thus for any individual seller, allowing the equilibrium not to satisfy the non-malicious condition potentially gives him a significant improvement in utility.

\begin{algorithm}
	\begin{enumerate}
		\item Initialize all prices to be infinity, initialize the price of seller $i$ to $p_i$.
		\item In Breath First Search (BFS) order starting from $i$, let each seller best response to the current prices. If the seller can only get $0$ utility, set his price to the value of the edge to the seller that is towards $i$ (this is well defined for trees, this is the seller that inserted him to the BFS queue).
	\end{enumerate}
	\caption{Algorithm to find an equilibrium in a tree with a seller $i$ pricing at $p_i$}
	\label{alg:tree-ne-utility}
\end{algorithm}

\begin{proposition}\label{prop:line-any-ne-util}
	Fix any tree, any seller $i$ in the tree, and any fixed $p_i$. Algorithm \ref{alg:tree-ne-utility} finds a Nash equilibrium in which seller $i$ prices at $p_i$ and has utility $p_i$ times the number of adjacent edges with value at most $p_i$.
\end{proposition}
\begin{proof}
	The argument that this algorithm terminates in an equilibrium is exactly the same as the argument for Algorithm \ref{alg:line-ne}. 	The way we break ties for sellers with $0$ utility ensures that not only the currently considered seller is best responding to all sellers for which the price is finite, but it also makes sure all such sellers are still best responding (as it leaves them no slack).
	
	Clearly the price of $i$ in this equilibrium is $p_i$. To complete the proof we observe that any neighbor of $i$ for which the edge with $i$ has value at most $p_i$, makes this edge tight (and it is sold), while any other neighbor prices high enough such that the edge is not sold and such that $i$ has no slack on that edge. This implies that the utility of $i$ is as claimed.
\end{proof}

\subsection{Total sellers utility}
We have seen that allowing arbitrary equilibrium might have windfall for a single seller. But what about the group of sellers as a whole? and the society?
We next show that both might gain a lot from allowing the equilibrium not to satisfy the non-malicious condition (proof is in Appendix \ref{app:malice}).


\begin{proposition}
	\label{prop:mal-better-ne}
	For any $d\geq 2$ there is a tree with maximum degree $d$ for which the revenue (and welfare)
	in equilibrium (and thus also the revenue of a monopolist seller) is factor $\Omega(\ln d)$ larger than the revenue (and welfare) in {\em any} non-malicious Nash equilibrium.
\end{proposition}

We remark that the above theorem also show that for non malicious NE, the logarithmic ``Price of Stability'' bound presented in Corollary \ref{cor:tree-lb} is essentially tight.


\bibliographystyle{plain}

\appendix

\section{Path Graphs: Proofs 
}
\label{app:lines}

We begin by a useful and simple observation about the price set by a best responding seller, it is used repeatedly throughout the paper.\footnote{To see this, note that if no edge is tight, a slight increase in price will increase the seller revenue (some edges must be sold as the seller's revenue is positive).
}
\begin{observation}
	\label{obs:tight-edge}
	For any graph and any vector of prices for the sellers, if seller $i$ is best responding and his revenue is positive, then there is an incident edge that is tight.
\end{observation}

\noindent{{\textbf{ Proof for Proposition \ref{prop:line-half}:}}\emph{
	For any path graph, Algorithm \ref{alg:line-ne} terminates in a NE after a linear number of steps. The total revenue in this equilibrium is at least half of the maximum welfare.
}

	\begin{proof}
		As pointed up above, the algorithm terminates in an equilibrium.
		We argue that one of the two runs attains  revenue of at least half the maximum welfare.
		
		Note that by this algorithm, by Observation \ref{obs:tight-edge} every sold edge is tight, thus the total revenue from that edge equals the value of the edge.
		Also observe that if an edge has value that is not smaller than the value of the previous edge, it must be sold. Additionally, the first edge is always sold.
		Consider the total value of all edges that are no-smaller than the previous edge, plus the first edge, in each of the two directions.
		Every edge is counted in one of these directions, except for strict local minima (i.e., edges that have strictly lower weight than its two neighbouring edges) that might not be counted. Edges that are weakly local maxima (i.e., have weights at least as high as their neighbouring edges) are counted in both directions.
		Now, pick one of the directions and match the first strict local minimum to the first edge, and from there, match each strict local minimum to the last local maximum prior to it.
		Clearly, the value of each strict local minimum is at most the value of the edge matched to it.
		It follows that the sum of sold edges in the two executions together is at least the sum of the weights in the path.
		Thus, one of the directions obtains at least half the total weight.
	\end{proof}
	
\section{Cycles and Trees: Proofs 
}
	\label{app:cycle}

In the first part of this section we prove Theorem \ref{thm:cycle}.

We first take care of two special cases.
First, if all edges have the same value, there is a fully-efficient fully-revenue-extracting non-malicious equilibrium in which every seller prices at half that value.
Second, if the cycle is a triangle ($n=3$) with edges $x\geq y\geq z$, there is a NE with price 0 at the node that is not on the edge $z$, and prices $x$ and $y$ at the other sellers on edges with values $x$ and $y$, respectively. This is a NE with at least $2/3$ of the welfare.

We next move to consider cycles of at least $4$ sellers when not all edges have the same value. We analyze the algorithm for cycle and prove Claim \ref{claim:cycle-app} which implies Theorem \ref{thm:cycle}.
We start with some supporting claims.

\begin{claim}
	After a complete cycle of the path-graph alg. on the cycle, the only seller who might not be best responding is the 2nd seller.
\end{claim}
\begin{proof}
	The claim is true by the properties of the path algorithm: each seller which sets his price by the path algorithm and for which the next seller's price was set by the path algorithm, is best responding. The second seller is the only one for which this is not the case (as he did not get to update his price after the first seller changed his price from $0$).
\end{proof}

\begin{claim}
	At the best response dynamics, when a seller that is best responding is changing his price, he makes the edge with his neighboring seller that changed his price before him tight (and sold), and that seller is now best responding (so the only seller that might not be best responding is his other neighbor, next in order).
\end{claim}
\begin{proof}
	Consider a seller that changed his price from $p$ to $p'\neq p$ to best respond. That seller must be selling at price $p'$, and one of his edges must be tight by Observation \ref{obs:tight-edge}. Consider the seller on the other edge. By induction he is the only seller that might not be best responding. Assume that his price was $q$ and that was a best response to the set of price before $p$ was changed to $p'$.
	Let us denote the value of the edge between the seller with price $p$ and the seller with price $q$ by $X$, the value of the other edge adjacent to $q$ by $Y$, and the price of the other seller on this edge by $t$.
	$q$ was  a best response when the neighboring prices were $p$ and $t$. By the path algorithm it holds that $p+q\geq X$ and $q+t\geq Y$.
	
	Assume that $q$ is changed to $q'$ to best respond to the change from $p$ to $p'$. If $q=q'$ we are in a Nash equilibrium.
	Otherwise the seller with price $q'$ must be selling at least one edge, and one edge must be tight. If $q'=X-p'$ we are done (as the edge with value $X$ is tight and $p'$ is tight on both sides so is a best response for the previous seller, thus the only seller that might not be best responding is the next seller).
	Otherwise, $q'=Y-t\geq 0$ and thus $Y\geq t$. By the path algorithm and the hypothesis about the best response dynamics so far, $t$ was chosen to make $t=Y-q$, thus $q=q'$ and we are in a Nash equilibrium.	
\end{proof}

	
	

	
	
\begin{claim}
	For any cycle, the best response process terminates after each seller  best-responds once.
\end{claim}
\begin{proof}
	For cycles with an even number of nodes: After a full cycle all edges are sold and are tight, this means that a change made by the first seller propagates with alternating signs, and the net change along a complete cycle is 0, which means that the seller is best responding without changing his price.

For cycles with an odd number of nodes:
	Assume that after a full cycle of best responses we are still not in a NE, and every edge is tight but the edge between seller one and two. Denote the current prices of seller 1,2 and 3 are $p,q$ and $t$, respectively.
	Assume that when $p$ is updated it increases by $\delta>0$ (if at this update it is decreased, we can consider the second third and fourth sellers instead). For the cycle to continue, $q$ must decrease by $\delta$ (as all edges are becoming tight when best responding). This means that $2(q-\delta)>q$ (as after the decrease he will gain $q-\delta$ from both edges, while without it the seller will only get $q$ from the next edge). For the cycle to continue next $t$ will increase to $t+\delta$. Now, after a full cycle, since the cycle is odd and all edges are tight, seller $1$ must update his price back to $p$, and at this point the current price of the next seller is $q-\delta$ while the following is $t+\delta$. For the algorithm to continue the second seller must prefer going back up to $q$, but this mean he will only sell the edge with $p$ and not both edges, so it has to be the case that $q>2(q-\delta)$, a contradiction.     	
\end{proof}

\begin{claim}
	\label{claim:cycle-app}
	For any cycle 
	the algorithm running either clockwise or counter-clockwise gets revenue that is at least one quarter
	of the maximum welfare.
\end{claim}
\begin{proof}
	In one of the executions of the path algorithm (either clockwise or counter-clockwise), the revenue is at least a half of the welfare (by Proposition \ref{prop:line-half}).
	If this is an equilibrium, by Proposition \ref{prop:line-half} the equilibrium revenue is at least half the welfare.
	Otherwise, the best response dynamics starts, and it only adds sold edges and all of them are made tight, and in particular the last edge not tight is sold and tight.
	In the last step, one edge that was originally sold might not be sold anymore. Denote the value of this edge by $Y$, and the value of the previous and next edges by $X$ and $Z$ respectively. We know that $X$ is sold. If $Y<X$, at most half the revenue is lost and we are done.
	We are left to consider the case that $Y\geq X$ and $Y$ is either not sold, or sold but not tight.
	
	If $Y\geq X$ and $Y$ is not sold, $Z$ must be sold.
	If the prices on the two ends of $Y$ are $p$ and $q$, $Y$ not sold means that $p+q>Y$, while we get revenue of at least $p+q$ from edges $X$ and $Z$ (which are sold). Thus at most half the revenue is lost.
	
	Finally, consider the case that $Y\geq X$ and $Y$ is being sold but not tight. In this case the revenue is at least $2(p+q)$. It also hold that $2p>Y-q$ and $2q>Y-p$ since each seller prefers the current prices over making $Y$ tight, and thus
	$2(p+q)>2Y-p-q$, therefore $(3/2) (p+q)>Y$. We conclude that $Y$ is smaller than the revenue, as the revenue is at least $2(p+q)$.
\end{proof}
	
	
	
\label{app:pos-proofs}

\noindent{\textbf{Proof for claim \ref{claim:log-approx-demand-curve}:}}
		\emph{For every node $v$ it holds that\\ $\sum_{(v,u)\in E^v} v_{(v,u)} \leq r_v \cdot 2(\ln d +1) $.}
	\begin{proof}
		The bound is derived  by observing that the $j$'th highest value edge in $E^v$ cannot have slack larger than $r_v/j$
		(and thus value of more than $r_v\cdot (2/j)$); otherwise for a small enough $\epsilon>0$, a price of $r_v/j+\epsilon$, will get $j$ edges to buy (since the price on the other side is at most half the edge value) and will result in revenue larger than $r_v$, thus a profitable deviation.
		Summing over all these edges from highest value to lowest, the total value is bounded from above by the harmonic sum times $r_v$, that is,
			$\sum_{(v,u)\in E^v} v_{(v,u)} \leq 2r_v \sum_{j=1}^{d} (1/j)\leq
			r_v \cdot 2(\ln d +1) $. 
	\end{proof}

\section{Impossibility Results:
Proofs}
\label{app:lower-bounds}

\noindent{Proof for Theorem \ref{thm:general-lb}:}\emph{
	There exists  a family of graphs with maximum degree $d$ and arboricity $w$ for which $w^2=O(\ln d)$ and the revenue that a monopolist seller can get is factor $\Omega(w)$ larger than the revenue (and welfare) in any Nash equilibrium. It terms of $d$, the factor can be as large as $\Omega(\sqrt{\ln d})$ when $w^2=\Theta(\ln d)$. 
}
\begin{proof}
	Consider a clique of size $w+1$ (arboricity of at least $w/2$)
	such that any node in the clique is connected to the ``harmonic gadget'': $d$ edges with values $1,1/2,1/3,1/4,...,1/d$.
	The value of an edge that connects any two nodes in the clique is $4$.
	See Figure \ref{fig:graph-lb} for illustration.

	Observe that a monopolist can get revenue $(w+1)\left(\sum_{i=1}^{d} 1/i\right)\geq (w+1)\ln d$ by pricing any node in the clique at $0$ and any other node in the $w+1$ harmonic gadgets at the value of the edge that connects the node to the clique. Thus
	the revenue of a monopolist that prices all the nodes is at least $(w+1)\ln d$.
	
	
	We observe that in any Nash equilibrium there is at most one node in the clique whose price is less or equal to $1/w$.
	Indeed, the payoff for such a node is at most $1 + w (1/w) = 2$. ($1$ from the gadget, and the rest from the clique).
	If there are two such nodes, then each node can gain more than $2$ by tightening the edge between them (as the other one prices at most at $1$).
	
	Thus, the revenue (and welfare) in any NE 
	is at most:	
	$4 (w+1)^2$ (from the clique), plus
	$\ln d +1$ (from the harmonic gadget connected to the only possible node with price at most $1/w$), plus
	$(w+1) (1+ \ln w )$ from the gadgets of the other nodes of the clique.
	
	Thus, the total revenue (and welfare) of any NE is at most $4(w+1)^2+\ln d+1 + (w+1)(1 + \ln w)$,
	while the monopolist revenue is at least $(w+1) \ln d$. The monopolist can also get revenue of at least $2w^2$ by pricing all nodes of the clique at 2. Thus the regime of interest would be the one in which $w^2\leq \ln d$.
	The ratio of the monopolist revenue to equilibrium revenue (and welfare) is at least
	$$ \frac{\max\{(w+1) \ln d, 2w^2\}}{4(w+1)^2+\ln d+1 + (w+1)(1 + \ln w)} $$	
	For $w^2\leq \ln d$ the ratio between the monopolist revenue and the revenue (and welfare) in a NE tends to $w$.	
Also, the degree of each node is at most $2d$ and the arboricity is at least  $w/2$.	
\end{proof}

\noindent{Proof for Theorem \ref{thm:general-tree-lb}:}\emph{
There exists a family of trees with maximum degree $d$ for which the revenue that a monopolist seller gets is factor $\Omega(\ln \ln d)$ larger than the revenue (and welfare) in any NE.
}

\vspace{-2mm} 

\begin{proof}
Fix an integer $m$ to be determined later.
Consider a path graph with $2m+1$ edges with the following values. The first edge has value $5$. For any $j=1,2,\ldots m$, given that edge $2j-1$ has value $v$, set the value of the edge $2j$ to be $2v+2$ and the value of edge $2j+1$ to be $2v+6$. Any node of even index is additionally connected to an ``harmonic gadget'' with $d-2$ spikes: $d-2$ edges with values $1,1/2,1/3,1/4,...,1/(d-2)$. See Figure \ref{fig:general-tree-lb-prices} for illustration.
%
%
	\begin{figure}[t]
		\center
		\framebox{\includegraphics[width=3in, height=1.2in]{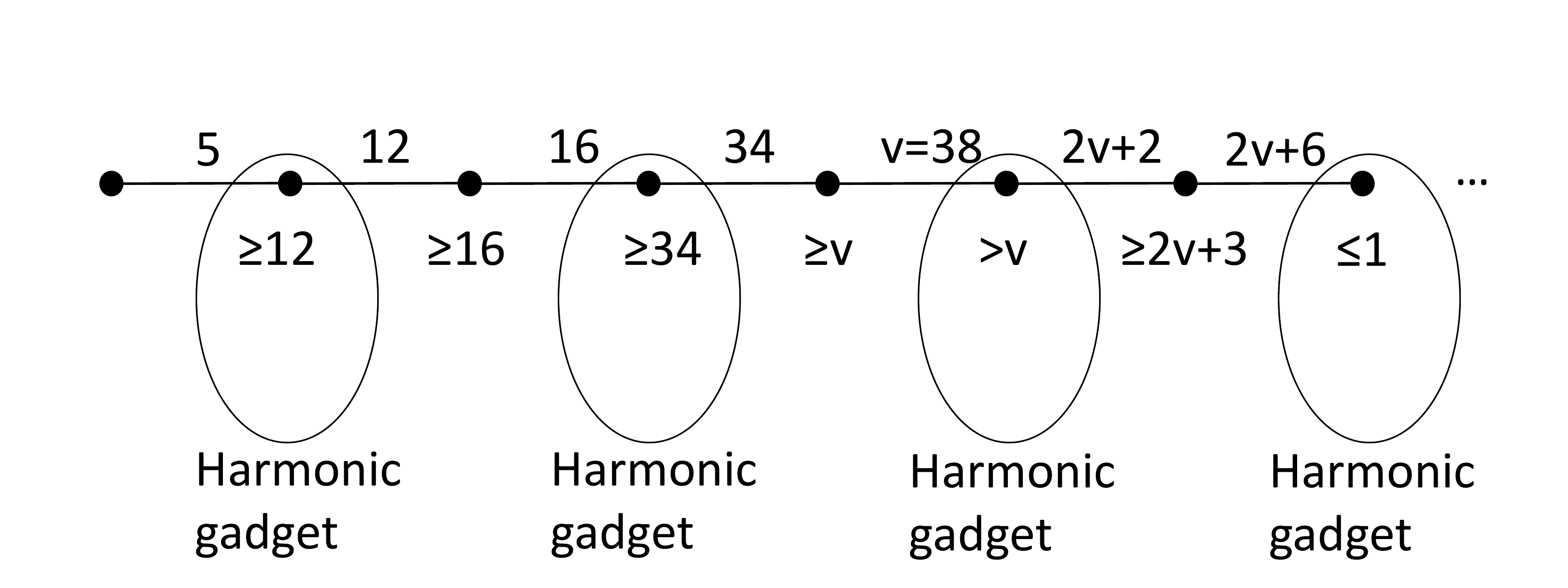}
		}
		\caption{
The construction for Theorem \ref{thm:general-tree-lb}. There is a path with $2m+1$ edges and values defined inductively from left to right.
The figure illustrates the argument that
there is at most one seller on the path with an even index that is pricing below $1$.
 For him to best responding below 1, all prior sellers are constrained to price high, and no other even seller prices below 1.
\label{fig:general-tree-lb-prices}}
	\end{figure}
	
We argue that in any Nash equilibrium there is at most a single seller of even index on the path with price that is at most $1$.
Consider the highest index seller on the path that prices at most at $1$, we argue that no seller of lower index on the path prices at most at $1$.
If this is seller 2, we are done. Otherwise, let this seller have index $j$ and define $v$ to be the value such that the edge previous to that seller to be of value $2v+6$.
We would like to argue that any even seller of lower index prices at least at the value of the edge prior to him, and thus prices at a price that is larger than $1$.

We use Figure \ref{fig:general-tree-lb-prices} for an illustration of the argument described next.
For seller $j$ with price of at most $1$, his utility is at most $3$.
For the previous seller (index $j-1$) to prevent $j$ from having a beneficial deviation in which his revenue is larger than $3$, seller $j-1$ must price at least at $2v+3$.
We now claim that seller $j-2$ price is larger than $v$, as if his price is at most $v$ then seller $j-1$ can price at $v+2$ and increase his utility to at least $2v+4>2v+3$.
Next we claim that to force seller $j-2$ to price at least at the value of the edge before him ($v$) then every seller on the path must price at least at the value of the edge after him.
This is so as seller $j-2$ has $0$ utility, and to make sure he does not decrease his price, the prior seller must price at least at the value of the edge after him ($v$).
Same argument hold for any prior seller, he has to price at least at the value of the edge after him, to prevent the next seller from decreasing his price.

Now, a monopolist seller can get revenue of at least $m\ln (d-2)$ by pricing every seller on the path at $0$ and gaining all the revenue from the $m$ harmonic gadgets.
He can also get revenue that is as high as the sum of all edges on the path, this revenue is at least $2^m$ as this is a trivial bound for the highest value edge on the path.
By the claim that there is at most one seller of even index on the path with price that is at most $1$, the total equilibrium revenue (and welfare) is bounded by the total welfare of the path, plus the revenue (and welfare) of the gadget connected to  the single seller on the path that prices at a price of at most $1$, which is at most $3+\ln d$.
The welfare of the path is at most twice the welfare of the odd index edges, as the value of the edges increases.
That sum of values of the odd index edges is $\sum_{i=1}^{m+1} a_i$ for the series defined by the recursion: $a_1=5$ and $a_{i+1}=2a_i+6$.
Now $\sum_{i=1}^{m+1} a_{i}  = a_1 + \sum_{i=1}^{m} a_{i+1} = a_1+ \sum_{i=1}^{m} (2a_i+6) = a_1 + 6m+2 \sum_{i=1}^{m} a_i $, thus
$a_{m+1}= a_1 + 6m+ \sum_{i=1}^{m} a_i $ or alternatively $\sum_{i=1}^{m+1} a_i = 2a_{m+1}-a_1-6m<2a_{m+1} $. Thus to bound the sum, we only need to bound $2a_{m+1}$.
We prove by induction that $a_i<6\cdot 2^i - 6$. For $i=1$ it indeed holds that $a_1=5<12-6=6$. Assume this holds for $i$, we prove it for $i+1$.
Indeed, by the induction hypothesis $a_{i+1}=2a_i+6<2(6\cdot 2^i -6)+6 =6\cdot 2 ^{i+1}-12+6=2^{i+1}-6$. Thus $a_{m+1}< 2^{m+1+1}-6 <2^{m+2}$ and therefore  $\sum_{i=1}^{m+1} a_i < 2^{m+3}$.
This implies that the welfare of the path graph is bounded by $2^{m+4}$.

We conclude that the ratio of the monopolist revenue to equilibrium revenue (and welfare) is at least
$ \frac{\max\{m \ln (d-2), 2^m\}}{3+\ln d + 2^{m+4}}. $
For $m=\Theta(\ln \ln d)$ this ratio tends to $\Omega(m)$ as we aimed to prove.
\end{proof}

\section{Trees: Best Response Dynamics}
\label{app:best-response-trees}

\noindent{Proof for Theorem \ref{thm:best-res-dyn-tree}:}\emph{
	In every tree, for every initial profile of prices, there exists a sequence of player best replies that terminates in a non-malicious Nash equilibrium.
}

%
\begin{proof}
	We will pick a leaf $u$ of the tree that is connected to the rest of the tree via vertex $v$, and the edge between $u$ and $v$ has, w.l.o.g. and for simplicity of notation, weight 1.
	Let $x_0$ be the initial price of $u$ and $y_0$
	be the initial price of $v$.  Our best reply (BR) dynamics will proceed  
 by repeatedly, for $i=1....$, let $x_i$ be $u$'s reply to $y_{i-1}$, and then recursively
	use a BR sequence that updates the rest of the tree, assuming that $u$'s value is set to $x_i$.  Note that this recursive BR sequence
	starts with $v$ updating his value (since all the other vertices are already best-replying from the previous recursive call, but then other vertices may
	update their value and $v$ may update again, and so on).  The value of $v$'s price at the end of the recursive call is called $y_i$.  The
	recursive call on the sub-tree terminates due to an inductive use of the theorem (i.e., the theorem is proved by induction on the number of vertices in
	the tree).  To ensure that the induction hypothesis applies to the recursive call which is applied not just to a subtree, but rather to a subtree to which an
	extra leaf $u$ with a {\em fixed} value $x_i$ is attached, we prove the theorem (inductively) also for trees in which each vertex may have an arbitrary
	number leaves with a fixed value attached to them.

	Before analyzing the dynamics of such a repeated loop of best replies in detail, we mention 
simple situations where the best-reply dynamics terminates
	since everyone is already best-replying:
The first cases are when $x_i = x_{i-1}$ or $y_i = y_{i-1}$.
The other cases are when 	$x_i+y_{i-1}=1$ or  $x_i+y_i=1$.
	The rest of the proof will thus assume, w.l.o.g., that none of these ever happen.  There is one other general case where we can immediately know that the BR dynamics ends: if the BR (either
	of $u$ or of $v$) increased the value from the previous one without changing the status of the $(u,v)$ edge.  I.e.:
	\begin{enumerate}
		\item If $x_{i-1} + y_{i-1} > 1$ and $x_i > x_{i-1}$ then $y_{i-1}$ is already a best-reply to $x_i$.  Similarly, if $x_i + y_{i-1} > 1$ and $y_i > y_{i-1}$ then $x_i$ is already a best reply to $y_i$.
		\item If $x_i + y_{i-1} < 1$ and $x_i > x_{i-1}$ then $y_{i-1}$ is already a best-reply to $x_i$.  Similarly, if $x_i + y_i < 1$ and $y_i > y_{i-1}$ then $x_i$ is already a BR to $y_i$.
	\end{enumerate}

	When analyzing the BR behavior of $u$ and $v$, we have tighter control on $u$ since it is just replying to $v$ (and the fixed values of the fixed leaves attached to it).  $v$, on the other hand,
	is best replying to $u$ as well as to the rest of the tree  connected to it; when $v$'s value changes, the values of the rest of the tree change which may cause further indirect changes. 
	So for $u$ we can state the following:
	

	\begin{enumerate}
		\item There exists a threshold $\theta$ and a price $x^*$ such that $u$'s BR to any $y > \theta$ is $x^*$ where $x^*+y>1$, while for any $y < \theta$ $u$'s BR $x$ to $y$ satisfies
		$x+y \le 1$.  Proof: Note that if $y'>y$ and the BR to $y$ already satisfies $x+y>1$ then $x$ must also be a BR to $y'$ so let $\theta$ be the infimum value of $y$ to which the BR satisfies $x+y>1$
		and let $x^*$ be that BR (we assume here some fixed tie breaking rule between multiple BR).
		\item For $y<\theta$, $u$'s BR $x$ to $y$ is monotone non-increasing in $y$.
		\item If $u$ has at most $d$ leaves with fixed value attached to it then there exists at most $d$ possible different values to the BR of $u$ to all values of $y$, in addition to best replies $x$ that satisfies
		exactly $x+y=1$.  Proof: at least one of the fixed value edges must be saturated.
	\end{enumerate}

	Let $x_i$ be a local minimum of the BR process, i.e. $x_i < x_{i-1}$ and $x_i < x_{i+1}$.
	Such a local minimum must exist since there are only finitely many possible values for $x_i$ (as we ruled out $x_i+y_{i-1}=1$).
	
	Case I: $x_i + y_{i-1} > 1$, and thus $x_i = x^*$.  Now there are two possibilities.
	(Ia) If $x_i + y_i > 1$ then since $x_{i+1}>x_i$ we also have $x_{i+1}+y_i>1$ a contradiction since in this case $x_i=x^*=x_{i+1}$.  Otherwise,
	(Ib) $x_i + y_i < 1$ and since if we had $x_{i+1}+y_i > 1$ it would again imply the contradiction that $x_i=x^*=x_{i+1}$ we must also have $x_{i+1}+y_i < 1$ and
	thus since $x_{i+1}>x_i$ we have that $y_i$ is already a BR to $x_{i+1}$.
	
	Case II: $x_i + y_{i-1} < 1$, there are three possibilities.  If (IIa) $x_i + y_i > 1$ then since $x_{i+1}>x_i$ we have $x_{i+1}+y_i > 1$ and we are done.
	Otherwise $x_i + y_i < 1$.  Now if (IIb) $y_i > y_{i-1}$ then $x_i$ is already a best reply to $y_i$ and we are done; otherwise (IIc) $y_i < y_{i-1} < \theta$ so
	$x_{i+1}+y_i < 1$ and thus $y_i$ is already a BR to $x_{i+1}$ and we are done.
\end{proof}

\section{Hyper-graphs}
In this section, we show how our positive result can be extended to 
hyper-graphs.
Hyper-graphs model single minded buyers, each is interested in a set of arbitrary size and is represented by an hyper-edge.
Like for graphs, for hyper-graphs it is also the case that the degree of a node is the number of hyper-edges incident to the node.
The {\em maximum degree} is the maximum degree of any node.
While for graphs all edges were incident on exactly two edges, this is no longer the case for hyper-graphs.
The size $|e|$ of an hyper-edge $e$ is defined to be the number of nodes incident to $e$.
Let $e_{max}$ denote the {\em size of the hyper-edge with maximum size}.

\noindent \emph{Arboricity in hyper graphs}:
Let the {\em incidence graph} of a hyper-graph be a bipartite graph with nodes on one side and hyper-edges on the other, and an edge between a node and an hyper-edge exists 
iff 
the node incidents the hyper-edge.
We use the definition of Berge \cite{berge1973graphs} for acyclicity of an hyper-graph: a hyper-graph is Berge-acyclic (acyclic for short) if its incidence graph is acyclic.
A {\em forest} is an acyclic hyper-graph.
 The {\em arboricity} of an undirected Hyper-graph is the minimum number of forests into which its Hyper-edges can be partitioned.
Equivalently, it is the minimum number of spanning forests needed to cover all Hyper-edges of the graph.


The proof of the next theorem is an adaptation of the proof of theorem \ref{thm:graph-poa}, and is repeated below with the required changes.

\begin{theorem}
	\label{thm:hyper-graph-pos}
	In every hyper-graph with maximum degree $d$ and arboricity $w$ and every non malicious equilibrium in it, the total revenue of all sellers is
	at least $\frac{1}{e_{max}(w+1+\log d)}$ 	
	fraction of the maximum welfare.
\end{theorem}
\begin{proof}
	Fix a hyper-graph and a non malicious NE in it.
	We say that an hyper-edge $e$ incident on vertex {\em $u$ has high slack for $u$} if the slack of  $u$ on $e$ is least $1/|e|$ of the value of
$e$. 
	We say that an hyper-edge $e$ is {\em all-high} if there is no vertex $u$ incident on $e$ for which $u$ has high slack on $e$ (such a hyper-edge clearly does not buy as it has at least $|e|$ incident sellers, each pricing more than $1/|e|$ fraction of
its the value).
	Let $E^H$ denote the set of all-high hyper-edges.
	We next partition the hyper-edges that are not all-high, mapping each hyper-edge to one of the incident sellers for which the edge has high slack. Let $E^u$ be the set of hyper-edges mapped to $u$.
	
	Observe that since every hyper-edge is either all-high or has high slack for some node, the set of all edges $E$ is covered as follows: $E=E^H \cup (\cup_{u} E^u)$.
	Thus $$\sum_{e\in E} v_{e }\leq  \sum_{e\in E^H} v_{e} + \sum_{u} \sum_{e\in E^u} v_{e} $$
	Denote the total revenue by $r$. To complete the proof we bound each of the two terms separately, the first by $w\cdot e_{max} \cdot r$ and the second by  $e_{max}\cdot (\ln d +1)\cdot r$,
	together proving the theorem.

	


We can prove, as in Claim  \ref{claim:log-approx-demand-curve}, that
for every node $u$, $\sum_{e\in E^u} v_{e} \leq r_u \cdot e_{max}\cdot (\ln d +1) $ ($r_u$ denotes the revenue of node $u$).
%
%
Thus,	$$\sum_{e\in E^u} v_{e} \leq \sum_{u} \left(	e_{max} \cdot r_u \cdot (\ln d +1) \right)  = e_{max} \cdot (\ln d +1)\cdot r.$$
We next bound the total value of the all-high hyper-edges.\footnote{We will use the fact that in a hyper-graph of arboricity $w$ there exists a mapping from hyper-edges to vertices such that every hyper-edge is mapped to one of its incident vertices and no vertex has more than $w$ hyper-edges mapped to it (this is true as we can just root each tree and map every hyper-edge to its parent node).}
	Since the hyper-edge is all-high and the equilibrium is non malicious, the price -- and thus also the revenue -- of each of the vertices on the hyper-edge $e$ is at least $1/|e|$ the value  of the hyper-edge $e$.
	Summing, again, over all vertices, we get that the
	total value of all all-high hyper-edges is at most $e_{max}\cdot w$ the total revenue of all vertices.
	Stated formally, consider the mapping $M$ from $E^H$ to the set of nodes $N$ that maps each hyper-edge to an adjacent node and never map more than $w$ hyper-edges to the same node.
	It holds that
	$$
	\sum_{e\in E^H} v_{e} \leq
	\sum_u \sum_{e\in M(u)} v_{e} \leq
	\sum_u \left(w\cdot e_{max}\cdot r_u \right) \leq e_{max}\cdot w\cdot r
	$$
\end{proof}


\begin{corollary}
	\label{cor:tree-lb-hyper}
	In any hyper-tree with maximum degree $d$ and every non malicious equilibrium in it, the total revenue of all sellers is $\Omega(1/(e_{max}\cdot \log d))$
	fraction of the maximum welfare.
\end{corollary}

\section{Proofs from Section \ref{sec:windfall}}
\label{app:malice}

\noindent \textbf{Proof of Proposition \ref{prop:mal-better-ne}}:

\begin{proof}
	For degree $d$ consider the following tree that is constructed from two stars connected by their centers with an edge of value $2$. The first star with a center denoted by $A$ had $d-1$ spikes, and the $i$-th edge has value $1/i$ (``harmonic gadget''). The second star with a center denoted by $B$ has $d-1$ spikes, each with value $3/(d-1)$.
	
	
	First observe that equilibrium revenue is at least $\ln d$, as the following is an equilibrium:
	$B$ sets a price of $2$, $A$ sets a price of $0$, and all the other neighbors of $B$ price at $10$, while the neighbors of $A$ each price at the value of the edge, getting revenue that equals the welfare of the harmonic gadget.
	
	We claim that in any non-malicious equilibrium $B$ must price at a price $p$ that is at most $3/(d-1)$. Assume that $p>3/(d-1)$. In this case the revenue of $B$ is at most $2$ (from the edge with $A$), and as it is a non-malicious equilibrium and all edges between $B$ and its spikes are not sold, each of these adjacent sellers must price at $0$. But in this case $B$ can deviate, price at $3/(d-1)$, and get revenue of at least $3$, a contradiction.

	As $B$ must price at a price that is at most $3/(d-1)$, in any non-malicious equilibrium $A$ gets revenue of at least $2-3/(d-1)$ by pricing at $2-3/(d-1)$.
	Note that for $d>7$ it holds that $2-3/(d-1)>3/2$.
	
	If $A$ set a price that is no larger than $1/2$, he get revenue that is at most $3/2$. We conclude that for $d>7$ the price of $A$ will be at least $1/2$.
	
	Finally, observe that in any case that $A$'s price is at least $1/2$, the welfare (and thus the revenue), is at most $7$, this is in contrast to the revenue in equilibrium, 
	which is at least $\ln d$.	
\end{proof}

%
%
%

\begin{thebibliography}{10}

\bibitem{BabaioffKP09}
Moshe Babaioff, Robert Kleinberg, and Christos~H. Papadimitriou.
\newblock Congestion games with malicious players.
\newblock {\em Games and Economic Behavior}, 67(1):22 -- 35, 2009.

\bibitem{BabaioffLN13}
Moshe Babaioff, Brendan Lucier, and Noam Nisan.
\newblock Bertrand networks.
\newblock In {\em ACM Conference on Electronic Commerce (ACM-EC)}, 2013.

\bibitem{Balcan06}
Maria-Florina Balcan and Avrim Blum.
\newblock Approximation algorithms and online mechanisms for item pricing.
\newblock In {\em Proceedings of the 7th ACM Conference on Electronic
  Commerce}, EC '06, pages 29--35, 2006.

\bibitem{berge1973graphs}
C.~Berge.
\newblock {\em Graphs and Hypergraphs}.
\newblock North-Holland mathematical library. Amsterdam, 1973.

\bibitem{Bert83}
Joseph Louis~Fran\c{c}ois Bertrand.
\newblock theorie mathematique de la richesse sociale.
\newblock {\em Journal de Savants}, 67:499–508, 1883.

\bibitem{CD11}
Yang Cai and Constantinos Daskalakis.
\newblock On minmax theorems for multiplayer games.
\newblock In {\em Proceedings of the Twenty-second Annual ACM-SIAM Symposium on
  Discrete Algorithms}, SODA '11, pages 217--234, 2011.

\bibitem{CDT09}
Xi~Chen, Xiaotie Deng, and Shang-Hua Teng.
\newblock Settling the complexity of computing two-player nash equilibria.
\newblock {\em J. ACM}, 56(3):14:1--14:57, 2009.

\bibitem{Cor04}
Margarida Corominas-Bosch.
\newblock Bargaining in a network of buyers and sellers.
\newblock {\em Journal of Economic Theory}, 115(1):35 -- 77, 2004.

\bibitem{Cournot1838}
Antoine~Augustin Cournot.
\newblock {\em Recherches sur les principes mathematiques de la theori des
  Richesses}.
\newblock 1838.

\bibitem{DFSS14}
Argyrios Deligkas, John Fearnley, Rahul Savani, and Paul Spirakis.
\newblock {\em Web and Internet Economics: 10th International Conference, WINE
  2014, Beijing, China, December 14-17, 2014. Proceedings}, chapter Computing
  Approximate Nash Equilibria in Polymatrix Games, pages 58--71.
\newblock 2014.

\bibitem{Ell66}
C.~Ellet.
\newblock {\em An essay on the laws of trade in reference to the works of
  internal improvement in the United States}.
\newblock Reprints of economic classics. A.M. Kelley, 1966.

\bibitem{GHKSW06}
Andrew~V. Goldberg, Jason~D. Hartline, Anna~R. Karlin, Michael Saks, and Andrew
  Wright.
\newblock Competitive auctions.
\newblock {\em Games and Economic Behavior}, 55(2):242 -- 269, 2006.

\bibitem{Guruswami05}
Venkatesan Guruswami, Jason~D. Hartline, Anna~R. Karlin, David Kempe, Claire
  Kenyon, and Frank McSherry.
\newblock On profit-maximizing envy-free pricing.
\newblock In {\em Proceedings of the Sixteenth Annual ACM-SIAM Symposium on
  Discrete Algorithms}, SODA '05, pages 1164--1173, 2005.

\bibitem{KKOPS05}
Sham~M. Kakade, Michael Kearns, Luis~E. Ortiz, Robin Pemantle, and Siddharth
  Suri.
\newblock Economic properties of social networks.
\newblock In {\em Advances in Neural Information Processing Systems 17}, pages
  633--640. MIT Press, 2005.

\bibitem{Kea07}
Michael Kearns.
\newblock {\em In {N}oam {N}isan, {T}im {R}oughgarden, {E}va {T}ardos and
  {V}ijay {V}azirani (Editors), Algorithmic Game Theory. Chapter 7. Graphical
  Games}.
\newblock Cambridge University Press., 2007.

\bibitem{KLS01}
Michael~J. Kearns, Michael~L. Littman, and Satinder~P. Singh.
\newblock Graphical models for game theory.
\newblock In {\em Proceedings of the 17th Conference in Uncertainty in
  Artificial Intelligence}, UAI '01, pages 253--260, 2001.

\bibitem{KT08}
Jon Kleinberg and \'{E}va Tardos.
\newblock Balanced outcomes in social exchange networks.
\newblock In {\em Proceedings of the Fortieth Annual ACM Symposium on Theory of
  Computing}, STOC '08, pages 295--304, 2008.

\bibitem{KM01}
E.~Kranton and Deborah~F. Minehart.
\newblock A theory of buyer-seller networks.
\newblock {\em American Economic Review}, 91:485--508, 2001.

\bibitem{KMR11}
Robert Krauthgamer, Aranyak Mehta, and Atri Rudra.
\newblock Pricing commodities.
\newblock {\em Theor. Comput. Sci.}, 412(7):602--613, 2011.

\bibitem{Lee15}
Euiwoong Lee.
\newblock Hardness of graph pricing through generalized max-dicut.
\newblock In {\em Proceedings of the Forty-Seventh Annual ACM on Symposium on
  Theory of Computing}, STOC '15, pages 391--399, 2015.

\bibitem{LBATL07}
E.w Lee, D.~Buchfuhrer, L.~Andrew, A.~Tang, and S.~Low.
\newblock Progress on pricing with peering.
\newblock In {\em Proceedings of the 45th annual Allerton Conference on
  Computing, Communications and Control}, Allerton ’07, pages 286--291, 2007.

\bibitem{LOS:J}
Daniel Lehmann, Liadan~Ita O'Callaghan, and Yoav Shoham.
\newblock Truth revelation in approximately efficient combinatorial auctions.
\newblock In {\em JACM 49(5)}, pages 577--602, Sept. 2002.

\bibitem{NW61}
C.~St.J.~A. Nash-Williams.
\newblock Edge-disjoint spanning trees of finite graphs.
\newblock {\em Journal of the London Mathematical Society}, s1-36(1):445--450,
  1961.

\end{thebibliography}

\end{document}